\documentclass{article}

\usepackage{amsthm}  %% to get \theroemstyle
\usepackage{amsmath, amssymb}
\usepackage{color}
\usepackage{colortbl}
\usepackage{fullpage}
\usepackage{algorithm}
\usepackage[noend]{algpseudocode}
\usepackage{thmtools, thm-restate}
\usepackage{tikz}
\usepackage{enumerate}
\usepackage{thm-restate}
\usepackage{hyperref}
\usepackage{graphicx}
\usepackage{amsopn}
\usepackage{caption}
\usepackage{hyperref}
\usepackage{subcaption}
\usepackage{enumitem,linegoal}
\usepackage{comment}
\usepackage{float}
\usepackage[super]{nth}
\usepackage{subcaption}

\usepackage{natbib}
\usepackage{mdframed}
\usepackage{thm-restate}

\newtheorem{theorem}{Theorem}[section]
\newtheorem{lemma}[theorem]{Lemma}

 \newtheorem{corollary}[theorem]{Corollary}

\theoremstyle{definition}
 \newtheorem{definition}[theorem]{Definition}
\newtheorem{example}[theorem]{Example}
\newtheorem{remark}[theorem]{Remark}
\newtheorem{observation}[theorem]{Observation}

\newcommand{\np}{\textsf{NP}}
\newcommand{\pls}{\mathsf{PLS}}

\newcommand{\agents}{\mathcal{N}}
\newcommand{\items}{\mathcal{M}}
\newcommand{\M}{\mathcal{M}}
\newcommand{\I}{\mathcal{I}}

\newcommand{\valus}{\mathcal{V}}

\newcommand{\alg}{\mathsf{ALG}}

\newcommand{\ef}{\mathsf{EF}}
\newcommand{\efx}{\mathsf{EFX}}

\newcommand{\mms}{\mathsf{MMS}}
\newcommand{\eefx}{\mathsf{EEFX}}
\newcommand{\prop}{\mathsf{PROP1}}

\newcommand{\apx}{\mathsf{APX}}

\newcommand{\idefx}{\mathsf{ID\text{-}EFX}}
\newcommand{\ideefx}{\mathsf{ID\text{-}EEFX}}

\newcommand{\todo}[1]{\textcolor{red}{TODO: #1}}

\definecolor{mygreen}{RGB}{20,140,80}
\definecolor{mylightgray}{RGB}{230,230,230}

\definecolor{mygreen}{RGB}{20,140,80}
\definecolor{mydarkgray}{gray}{0.15} 
\definecolor{oceanblue}{HTML}{2c55c2}
\hypersetup{
     colorlinks=true,
     citecolor= mygreen,
     linkcolor= black
}

\usepackage[preprint]{neurips_2024}

\usepackage[utf8]{inputenc} % allow utf-8 input
\usepackage[T1]{fontenc}    % use 8-bit T1 fonts
\usepackage{hyperref}       % hyperlinks
\usepackage{url}            % simple URL typesetting
\usepackage{booktabs}       % professional-quality tables
\usepackage{amsfonts}       % blackboard math symbols
\usepackage{nicefrac}       % compact symbols for 1/2, etc.
\usepackage{microtype}      % microtypography
\usepackage{xcolor}         % colors

\title{Epistemic EFX Allocations Exist for Monotone Valuations}

\author{%
    Hannaneh Akrami\\
   Max Planck Institute for Informatics\\
   Graduiertenschule Informatik, 
   Universit\"at des Saarlandes\\
   \texttt{hakrami@mpi-inf.mpg.de} \\
   \And
   Nidhi Rathi\\
   Max Planck Institute for Informatics\\
   \texttt{nrathi@mpi-inf.mpg.de}
}

\begin{document}

\maketitle
\begin{abstract}
    We study the fundamental problem of \emph{fairly} dividing a set of indivisible items among agents with (general) monotone valuations. The notion of \emph{envy-freeness up to any item} ($\efx$) is considered to be one of the most fascinating fairness concepts in this line of work. Unfortunately, despite significant efforts, existence of $\efx$ allocations is a major open problem in fair division, thereby making the study of approximations and relaxations of $\efx$ a natural line of research. Recently, \cite{Caragiannis2023} introduced a promising relaxation of $\efx$, called \emph{epistemic} $\efx$ ($\eefx$). We say an allocation to be $\eefx$ if, for every agent, it is possible to shuffle the items in the remaining bundles so that she becomes ``$\efx$-satisfied''. \cite{Caragiannis2023} prove existence and polynomial-time computability of $\eefx$ allocations for additive valuations. A natural question asks what happens when we consider valuations more general than additive?

    We address this important open question and answer it affirmatively by establishing the \emph{existence} of $\eefx$ allocations for an arbitrary number of agents with general \emph{monotone} valuations.  To the best of our knowledge, $\eefx$ is the only known relaxation of $\efx$ (beside $\ef1$) to have such strong existential guarantees. Furthermore, we complement our existential result by proving computational and information-theoretic lower bounds. We prove that even for an arbitrary number of (more than one) agents with identical submodular valuations, it is $\pls$-hard to compute $\eefx$ allocations and it requires exponentially-many value queries to do so.
\end{abstract}

\section{Introduction}\label{sec:intro}

The theory of fair division addresses the fundamental problem dividing a set of resources in a \emph{fair} manner among individuals (often called as agents) with varied preferences. This problem arises naturally in many real-world settings, such as division of inheritance, dissolution of business partnerships, divorce settlements, assigning computational resources in a cloud computing environment, course assignments, allocation of radio and television spectrum, air traffic management, course assignments, to name a few \citep{etkin2007spectrum,moulin2004fair,vossen2002fair,budish2012multi,pratt1990fair}. Although the roots of fair division can be found in  antiquity, for instance, in ancient Greek mythology and the Bible, its first mathematical exposition dates back to the seminal work of  Steinhaus, Banach, and Knaster~\citep{steinhaus1948problem}. Since then, the theory of fair division has received significant attention and a flourishing flow of research from areas across economics, social science, mathematics, and computer science; see \cite{survey2022,brams1996fair,brandt2016handbook,robertson1998cake} for excellent expositions. 

 The development of fair division protocols plays a crucial role in ensuring equitable outcomes in the design of many social institutions. With the advent of internet, the necessity of having division rules that are both transparent and agreeable or, in other words, \emph{fair} has become evident \citep{moulin2019fair}. There are many examples to see how the principles of fair division are being applied in various technological platforms today~\citep{AdjustedWinner,spliddit}.
%For example, \cite{AdjustedWinner} fairness/non-discrimination in generative AI models, dating sites, \todo{...} 
%Thus, concepts of fairness and their interaction with efficiency became relevant to the computer science (CS) community.

   \begin{table}
    \caption{The additive valuation functions of $3$ agents for $7$ goods.}\label{table}
    \begin{center}
        \begin{tabular}{ c | c c c c c c c c} 
                  & $g_1$ & $g_2$ & $g_3$ & $g_4$ & $g_5$ & $g_6$ & $g_7$ \\ \midrule
	        $v_1$ & $100$ & $100$ & $100$ & $1$ & $1$ & $1$ & $1$\\  
 	        $v_2$ & $1$ & $1$ & $1$ & $100$ & $100$ & $100$ & $1$\\  
            $v_3$ & $1$ & $50$ & $50$ & $1$ & $1$ & $1$ & $55$
        \end{tabular}
    \end{center}
    \end{table}

Some of the central solution concepts and axiomatic characterizations in the fair-division literature stem from the cake-cutting context~\citep{moulin2004fair} where the resource to be divided is considered to be a (divisible) cake $[0,1]$. The quintessential notion of fairness---\emph{envy-freeness}---was also mathematically formalized in this setup \citep{foley1967resource, varian1973equity}. We say an allocation is $\emph{envy-free}$ if every agent prefers their share in the division at least as much as any other agent's share. Strong existential guarantees of envy-free cake division that also establishes a connection with topology~\citep{stromquist1980cut,edward1999rental}, has undeniably made envy-freeness as the representative notion of fairness in resource-allocation settings.  Unfortunately, an envy-free allocation is not guaranteed to exist when we need to fairly divide a set of indivisible items: consider two agents and a single item: only one agent can get the item, and the other player will be envious. Furthermore, it is $\np$-hard to decide whether an envy-free allocation exists e.g., see~\cite{bouveret2016characterizing}. Infeasibility along-with high computational complexity of envy-free allocations has led to study of its various relaxations for discrete setting. 

In this paper, we consider the setting where the resource is a set of discrete or indivisible items, each of which must be wholly allocated to a single agent. A fair division instance consists of a set $\agents = \{1,2 \dots, n\}$ of $n$ agents and a set $\items$ of items. Every agent $i$ specifies her preferences via a valuation function  $v_i \colon 2^{\items} \to \mathbb{R}$. We study general \emph{monotone valuations} that pertains adding a good to a bundle cannot make it worse. The goal is to find a partition $X = (X_1, \dots, X_n)$ of the items where every agent $i \in \agents$ upon receiving bundle $X_i$ considers $X$ to be $\emph{fair}$. 

{\bf \boldmath Envy-freeness up to any item ($\efx$)}: One of the most compelling notions of fairness for discrete setting is \emph{envy-freeness up to any item} ($\efx$). This notion was introduced by \cite{caragiannis2016unreasonable}. We say an allocation is $\efx$ if every agent prefers her own bundle to the bundle of any other agent, after removing \emph{any} item from the latter. $\efx$ is considered to be the ``closest analogue of envy-freeness'' for discrete setting~\citep{caragiannis2019envy}. Unfortunately, despite significant efforts over the past few years, existence of $\efx$ allocations remain as the biggest and the most challenging open problem in fair division, even for instances with more than three agents with additive valuations~\citep{procaccia2020technical}. %Positive results are known only for very restricted instances. See Section \ref{sec:related}
.%; see~\cite{plaut2020almost,chaudhury2020efx,halpern2020fair}. 

%see~\cite{plaut2020almost,chaudhury2020efx,halpern2020fair}. Several approximations \cite{chaudhury2021little,mahara2021extension,amanatidis2020multiple} and relaxations \cite{amanatidis2021maximum,caragiannis2019envy,berger2021almost} of EFX have become an important line of research in discrete fair division. 

{\bf \boldmath Epistemic envy-freeness up to any item ($\eefx$):} A recent work of \cite{Caragiannis2023} introduced a promising relaxation of $\efx$, called as \emph{epistemic} $\efx$ (which adapts
the concepts of \emph{epistemic envy-freeness} defined by \cite{ABCGL18}). We call an allocation $X$ as  $\eefx$ if for every agent $i \in [n]$, there exists an allocation $Y$ such that $Y_i=X_i$ and for every bundle $Y_j \in Y$, we have $v_i(X_i) \geq v_i(Y_j \setminus g)$ for every $g \in Y_j$. That is, an allocation is $\eefx$ if, for every agent, it is possible to shuffle the items in the remaining bundles so that she becomes ``$\efx$-\emph{satisfied}''. See Example \ref{ex:efxvseefx} for a better intuition. 

\begin{example}\label{ex:efxvseefx}
    Consider a fair division instance consisting of $7$ items and $3$ agents with additive valuations as described in Table \ref{table}. Now consider the allocation $X$ where $X_1 =\{g_1, g_2, g_4\}, X_2=\{g_3, g_5, g_6\},$ and $X_3=\{g_7\}$. Note that $X$ is envy-free, and hence, $\efx$ and $\eefx$. Now assume that agent $1$ and $2$ exchange the items $g_3$ and $g_4$. Formally, let $Y = (\{g_1, g_2, g_3\}, \{g_4, g_5, g_6\}, \{g_7\})$. For $i \in \{1,2\}$, have $v_i(Y_i) = 300 > 201 = v_i(X_i)$, and $v_3(Y_3)=v_3(X_3)$. Therefore, intuitively it seems that $Y$ is a better allocation compared to $X$ since agents $1$ and $2$ are strictly better off and agent $3$ is as happy as before (i.e., $Y$ Pareto dominates $X$). However, note that while allocation $Y$ is still $\eefx$, it is not $\efx$. Namely, agent $3$ strongly envies agent $1$: $v_3(Y_1 \setminus \{g_1\}) = 100 > 55 = v_3(Y_3)$.
\end{example}
\cite{Caragiannis2023} establish existence and polynomial-time computability of $\eefx$ allocations for an arbitrary number of agents with a restricted class of \emph{additive} valuations. Thus, the following question naturally arises:

\smallskip
\begin{mdframed}
    Do $\eefx$ allocations exist for an arbitrary number of agents with general $\emph{monotone}$ valuations?
\end{mdframed}
\smallskip

%\smallskip
\subsection{Our Results}
We answer the above question in the affirmative and establish computational hardness and information-theoretic lower bounds for finding $\eefx$ allocations:
\begin{enumerate}
    \item $\eefx$ allocations are guaranteed to \emph{exist} for any fair division instance with an arbitrary number of agents having general \emph{monotone} valuations; see Theorem~\ref{thm:eefx}.
    \item An exponential number of valuation queries is required by any deterministic algorithm to compute an $\eefx$ allocation for fair division instances with an arbitrary number of agents with identical submodular valuations; see Theorem~\ref{thm:exp}.
    \item The problem of computing $\eefx$ allocations for fair division instances with an arbitrary number of agents having identical submodular valuations is $\pls$-hard; see Theorem~\ref{thm:pls}.
\end{enumerate}

It is relevant to note that, with the above results, the notion of $\emph{epsitemic}$-$\efx$ becomes the \emph{second} known relaxation of $\efx$ (beside $\ef1$), that admits such strong existential guarantees. Along-with its hardness results, the notion of $\eefx$ for discrete settings seem to enjoy results of \emph{similar} flavor as that of envy-freeness for cake division~\citep{stromquist1980cut,stromquist2008envy,deng2012algorithmic}. 

Similar computational hardness and information-theoretic lower bounds are known for computing an $\efx$ allocation between two agents with identical submodular valuations; see~\cite{plaut2020almost} and \cite{Goldberg2023}. We reduce our problem of computing an $\eefx$ allocation among an arbitrary number of agents with identical submodular valuations to the above computational problem for $\efx$. See Appendix \ref{app:pls} for further discussion on the $\pls$ class~\citep{johnson1988easy}.

Therefore, note that, similar computational hardness and information-theoretic bounds hold true for finding $\efx$ and $\eefx$ allocations. But, our work has proved a \emph{stark contrast} to $\efx$ by establishing guaranteed existence of $\eefx$ allocations for an arbitrary number of agents with monotone valuations, whereas existence of $\efx$ allocations for more than three agents even with additive valuations remain a major open problem.

\subsection{Our Techniques} \label{sec:technique}
We develop a novel technique to prove existence of $\eefx$ allocations for monotone valuations. In this section, we describe our technique in generality since which we believe it could be employed in other (fair division) problems as well. 

Consider a fair division instance\footnote{A fair division instance $\I = (\agents, \M, \valus)$ consists of a set $\agents$ of $n$ agents and a set $\valus$ consisting of agent-valuations over a set $\items$ of items.} $\I = (\agents, \M, \valus)$ and a \emph{desirable property} $\mathcal{P}$ of a bundle $B \subseteq \items$ for an agent $i \in \agents$. For example, in this work, we consider the fairness property of whether $B$ is $n$-\emph{epistemic}-$\efx$ for an agent $i$ (see Definition \ref{def:k-epistemic}). We say $B$ is \emph{desirable} to $i$ when $B$ satisfies the property $\mathcal{P}$ for agent $i$.
The goal is to find an allocation $A= (A_1, \ldots, A_n)$ such that $A_i$ is desirable to each agent $i \in \agents$; we call such an allocation as $\emph{desirable}$.

For any partitioning of the items into $n$ bundles $X_1, X_2, \ldots, X_n$, let us consider a bipartite graph $G(X)$ with one side representing the $n$ agents and the other side representing the $n$ bundles. There exists an edge $(i,j)$ between (the node corresponding to) agent $i$ and (the node corresponding to) bundle $X_j$, if and only if, bundle $X_j$ is \emph{desirable} to agent $i$. For any subset of the nodes $S \subseteq \agents$, let us write $N(S)$ to denote the set of all neighbours of $S$ in $G(X)$. 

Note that, if $G(X)$ has a perfect matching, then this matching translates to a \emph{desirable} allocation in $\I$. Therefore, let us assume that $G(X)$ does not admit a perfect matching and hence admits a Hall's violator set. That is, there exists a subset of agents $\{a_1, \ldots, a_{t+1}\}$ for which $N(\{a_1, \ldots, a_{t+1}\}) \leq t$. But also, there exists a subset of bundles $\{X_{j_1}, \ldots, X_{j_{k+1}}\}$ for which $N(\{X_{j_1}, \ldots, X_{j_{k+1}}\}) \leq k$. Let us assume that $\{X_{j_1}, \ldots, X_{j_{k+1}}\}$ is minimal. If $k \geq 1$, this means that we can find a non-empty matching of $((i_1, X_{j_1}), \ldots, (i_k, X_{j_k}))$ such that there exists no edge between agent $i \in \agents \setminus \{i_1, \ldots, i_k\}$ and  bundles $X_{j_1}, \ldots, X_{j_{k+1}}$. In other words, for all $\ell \in [k]$, $X_{j_\ell}$ is desirable to $i_\ell$ and is not desirable to any  $i \notin \{i_1, \ldots, i_k\}$. 

%Although it is a simple observation, to the best of our knowledge, studying the Hall's condition for bundles (instead of agents) and computing such a matching has not been used in previous works.

After finding such a matching, it is intuitive to allocate $X_{j_\ell}$ to $i_{\ell}$ for $\ell \in [k]$ and then recursively find a desired allocation of the remaining goods to the remaining agents. In order to do so, we need to ensure two important conditions.
\begin{enumerate}
    \item We can find a non-empty matching $((i_1, X_{j_1}), \ldots, (i_k, X_{j_k}))$ in each step.
    \item After removing $\{X_{j_1}, \ldots, X_{j_k}\}$ from $\M$, we can still find desirable bundles (with respect to the original instance) for the remaining agents.
\end{enumerate}

Whether ensuring these conditions is possible or not, depends on the property $P$. In this work, we prove this approach works when the property $P$ is $n$-epistemic-$\efx$, and thereby proving the existence of $\eefx$ allocations for monotone valuations. 

Although these two conditions might seem inconsequential, we prove that a stronger condition can simultaneously imply both of them. Namely, we only need to prove that at each step with $n'$ remaining agents, for any remaining agent $i$, we can partition the remaining items into $n'$ many bundles $X_1, \ldots, X_{n'}$ such that $X_j$ is desirable to $i$ for all $j \in [n']$. This way, at each step, we can ask one of the remaining agents to partition the remaining goods into $n'$ many desirable bundles with respect to her own valuation. Then, we either find a perfect matching, or we find a non-empty matching and reduce the size of the instance. 

A similar technique was also developed independently by \cite{bu2024fair} for finding $\prop$ allocations\footnote{$\prop$ requires each agent’s proportionality if one item is (hypothetically) added to that agent's bundle.} among agents with additive valuations in a \emph{comparison-based model}. Here, two bundles are presented to an agent and she responds by telling which bundle she prefers.

%In fact, once one has this approach in mind, the proof becomes very simple (see section \ref{sec:good-eefx}).
\subsection{Further Related Work}\label{sec:related}
%Positive results on $\efx$ are known for very restricted settings. 
\cite{plaut2020almost} proved the existence of $\efx$ for two agents with monotone valuations. For three agents, a series of works proved the existence of $\efx$ allocations when agents have additive valuations~\citep{chaudhury2020efx}, \emph{nice-cancelable} valuations~\citep{berger2021almost}, and finally when two agents have monotone valuations and one has an \emph{$\mms$-feasible} valuation~\citep{AkramiACGMM23}. $\efx$ allocations exist when agents have identical~\citep{plaut2020almost}, binary~\citep{halpern2020fair}, or bi-valued~\citep{amanatidis2021maximum} valuations.
Several approximations~\citep{chaudhury2021little,amanatidis2020multiple,chan2019maximin,farhadi2021almost} and relaxations~\citep{amanatidis2021maximum,caragiannis2019envy,berger2021almost,mahara2021extension,chasmjahan23,aram22,ef2x} of $\efx$ have become an important line of research in discrete fair division.

%The notion of epistemic envy-freeness was first introduced by Aziz et al. \cite{ABCGL18}.

Another relaxation of envy-freeness proposed in discrete fair division literature is that of {\em envy-freeness up to some item} ($\ef$1), introduced by \cite{budish2011combinatorial}. It requires that each agent prefers her own bundle to the bundle of any other agent, after removing some item from the latter. $\ef$1 allocations always exist and can be computed efficiently~\citep{lipton2004approximately}.

\emph{Proportionality}~\citep{dubins1961cut,steinhaus1948problem} is another well-studied notion of fairness having its roots in cake division literature. We say an allocation is proportional if each agent gets a bundle of items for which her value exceeds her total value for all items divided by the number of agents. It is easy to see that a proportional division cannot exist for the setting of discrete items. 

Among the relaxations of proportionality, the one that has received the lion's share of attention uses the so-called {\em maximin fair share} ($\mms$), i.e., the maximum value an agent can attain in any allocation where she is assigned her least preferred bundle, as threshold. Surprisingly, \cite{kurokawa2016can} proved that $\mms$ allocations may not always exist. Since then, research has focused on computing allocations that approximate $\mms$; e.g., see~\cite{amanatidis2017approximation,kurokawa2018fair,ghodsi2018fair,barman2020approximation,garg2020improved,FST21,simple,akrami2024breaking} for additive,~\cite{barman2020approximation, ghodsi2018fair,uziahu2023fair} for submodular,~\cite{ghodsi2018fair,seddighin2022improved,MMS-XOS} for XOS, and~\cite{ghodsi2018fair, seddighin2022improved} for subadditive valuations. 

Proportionality up to one good (PROP1)~\citep{conitzer2017fair} is another relaxation of proportionality which can be guaranteed together with Pareto optimality~\citep{barman2019proximity}. Proportionality up to any good (PROPX) on the other hand, is not a feasible notion in the goods setting~\citep{aziz2020polynomial}.

 %Unfortunately, it seems that $\ef$1 has moved way too far and has lost the fairness properties of envy-freeness.

An excellent recent survey by~\cite{survey2022} discusses the above fairness concepts and many more.
Another aspect of discrete fair division which has garnered an extensive research is when the items that needs to be divided are \emph{chores}. We refer the readers to the survey by~\cite{guo2023survey} for a comprehensive discussion.

\subsection{Organization:} We begin by discussing the preliminaries in Section~\ref{sec:prelim}. We prove our key result of guaranteed existence of $\eefx$ allocations for monotone valuations in Section~\ref{sec:good-eefx}. We conclude by proving information/theoretic lower bounds for computing an $\eefx$ allocation in Section~\ref{sec:hardness}. Towards the end, we discuss a list of many interesting open problems motivated by this work in Section~\ref{sec:conc}.
\section{Definitions and Notation} \label{sec:prelim}
For any positive integer $k$, we use $[k]$ to denote the set $\{1,2,\ldots,k\}$.
We denote a fair division instance by $\I = (\agents, \items, \valus)$, where $\agents=[n]$ is a set of $n$ agents, $\items$ is a set of $m$ items and $\valus=(v_1,v_2, \ldots, v_n)$ is a vector of valuation functions. For any agent $i \in \agents$, we write $v_i: 2^\items \rightarrow \mathbb{R}_{\geq 0}$ to denote her valuation function over the set of items.  For all $i \in \agents$, we assume $v_i$ is normalized; i.e., $v_i(\emptyset)=0$, and \emph{monotone}; i.e., for all $i \in \agents$, $g \in \items$ and $S \subset \items$, $v_i(S \cup \{g\}) \geq v_i(S)$.
Moreover, we say a valuation function $v$ is submodular if for all subsets $S,T \subseteq \items$, we have $v(S \cap T)+v(S \cup T) \leq v(S) + v(T)$ and we say $v$ is additive when $v(S) = \sum_{g\in S}v(g)$ for all subsets $S \subseteq \items$. Note that, monotone valuations are submodular and additive, making them the most general class of valuations.

For simplicity, we sometimes use $g$ instead of $\{g\}$ to denote an item $g \in \items$. We use ``items'' and ``goods'' interchangeably.

%In this work, we assume for all $i \in \agents$, $v_i$ is monotone. I.e., for all $i \in \agents$, either
%\begin{itemize}
%    \item for all $g \in \items$ and $S \subset \items$, $v_i(S \cup g) \geq v_i(S)$, or 
%    \item for all $c \in \items$ and $S \subset \items$, $v_i(S \cup c) \leq v_i(S)$.
%\end{itemize}
%In the first case, we call $i$ a ``good'' agents and otherwise, a ``chore'' agent. Note that good agents consider $\items$ as a set of positively valued items (i.e. goods) and chore agents consider $\items$ as a set of negatively valued items (i.e. chores). When analysing good agents, we sometimes refer to items as goods and when analysing chore agents, we sometimes refer to items as chores. 
%In order to have a unified notation for good and chore agents, for all $\oplus \in \{<, >, \leq, \geq, =\}$ we write $A \oplus_i B$ if $v_i(A) \oplus v_i(B)$ for good agents, and $v_i(B) \oplus v_i(A)$ for chores agent. Note that $A >_i B$ means that agent $i$ prefers bundle $A$ over bundle $B$ whether she is a good or a chore agent.

An allocation $X=(X_1, X_2, \ldots, X_n)$ of the items among agents is a partition of items into $n$ bundles such that bundle $X_i$ is allocated to agent $i$. That is, we have $X_i \cap X_j = \emptyset$ for all $i,j \in \agents$ and $\cup_{i \in [n]} X_i=\items$.

%Intuitively, good agents want to increase the value of their bundle and chore agents want to decrease it.

%\begin{definition}[Strong Envy]
%    A good agent $i$ upon receiving a bundle $A \subseteq \items$ strongly envies a bundle $B \in \items$, if there exists a good $g \in B$ such that $v_i(A) < v_i(B \setminus g)$. Similarly, a chore agent $i$ upon receiving a bundle $A \subseteq \items$ strongly envies a bundle $B \in \items$, there exists a chore $c \in A$ such that $v_i(A \setminus c) > v_i(B)$. Under an allocation $X$, we say agent $i$ strongly envies agent $j$, if upon receiving $X_i$, $i$ strongly envies $X_j$.
%\end{definition}

Let us now define the concept of \emph{strong envy} that characterizes one of most compelling notions of fairness in the literature - \emph{envy-freeness up to any item} ($\efx$).  
\begin{definition}[Strong Envy]
     For a fair division instance, we say an agent $i$ upon receiving a bundle $A \subseteq \items$ \emph{strongly envies} a bundle $B \in \items$, if there exists an item $g \in B$ such that $v_i(A) < v_i(B \setminus g)$. Under an allocation $X$, we say agent $i$ \emph{strongly envies} agent $j$, if upon receiving $X_i$, agent $i$ strongly envies the bundle $X_j$.
\end{definition}

%\begin{observation}
%    Upon receiving a bundle $A$, an agent $i$ strongly envies a bundle $B$, if and only if $A <_i B \setminus g$ where $g = \arg \max_{g' \in B} v_i(B \setminus g')$.
%\end{observation}
%\begin{observation}
%    Upon receiving a bundle $A$, a chore agent $i$ strongly envies a bundle $B$, if and only if $A \setminus c <_i B$ where $c = \arg \max_{c' \in A} v_i(A \setminus c')$.
%\end{observation}
%\begin{definition}[EFX]
%    An allocation $X=(X_1,X_2, \ldots, X_n)$ is ``envy-free up to any item'' or ``EFX'', if no agent strongly envies another agent. I.e., for all agents $i$ and $j$, $X_i \geq_i X_j \setminus g$ for all $g \in X_j$ if $i$ is a good agent, and $X_i \setminus c \geq_i X_j$ for all $c \in X_i$ if $i$ is a chore agent.
%\end{definition}

\begin{definition}[$\efx$]
    For a fair division instance, an allocation $X=(X_1,X_2, \ldots, X_n)$ is said to be ``\emph{envy-free up to any item}'' or ``$\efx$'', if no agent strongly envies another agent. i.e., for all agents $i$ and $j$, $v_i(X_i) \geq v_i(X_j \setminus g)$ for all $g \in X_j$.
\end{definition}

Recently,~\cite{Caragiannis2023} introduced a promising new notion of fairness --- \emph{epistemic $\efx$} -- by relaxing $\efx$, that we define next. They proved \emph{epistemic $\efx$} allocations among an arbitrary number of agents with additive valuations can be computed in polynomial time.

\begin{definition}\label{def:k-epistemic}
    For any integer $k$, agent $i \in [n]$ and subset of items $S \subseteq \items$, we say that a bundle $A \in S$ is ``$k$\emph{-epistemic-$\efx$}'' for $i$ with respect to $S$, if there exists a partitioning of $S \setminus A$ into $k-1$ bundles $C_1, C_2, \ldots, C_{k-1}$, such that for all $j \in [k-1]$, upon receiving $A$, $i$ would not strongly envy $C_j$. We call $C=\{C_1, C_2, \ldots, C_{k-1}\}$ a ``$k$\emph{-certificate}'' of $A$ for $i$ under $S$. Also we define 
    \begin{align*}
        \eefx^k_i(S) = \{\text{$A \subseteq S \ \vert \ A$ is ``$k$-epistemic-$\efx$'' for agent $i$ with respect to }S\}.
    \end{align*}
\end{definition}

\begin{definition}[$\eefx$]
    For a fair division instance, an allocation $X=(X_1,X_2, \ldots, X_n)$ is said to be \emph{epistemic $\efx$} or \emph{$\eefx$} if for all agents $i$, $X_i \in \eefx^n_i(\items)$.
\end{definition}

Note that the set of $\efx$ and $\eefx$ allocations coincide for the case of two agents.
Next, we define a notion of $\eefx$\emph{-graph} that plays a crucial role in proving the existence of $\eefx$ allocations.

\begin{definition}
    For a fair division instance, consider a partition of $\items$ into $n$ bundles $Y_1, \ldots, Y_n$. We define the $\eefx$\emph{-graph} as an undirected bipartite graph $G=(V,E)$, where $V$ has one part consisting of $n$ nodes corresponding to the agents and another part with $n$ nodes corresponding to the bundles $Y_1, \ldots, Y_n$. There exists an edge $(i,j)$ between (the node corresponding to) agent $i$ and (the node corresponding to) bundle $Y_j$ if and only if $Y_j \in \eefx^n_i(\items)$.
\end{definition}

We abuse the notation and refer to the ``nodes corresponding to agents'' as ``agents'' and also refer to the ``nodes corresponding to bundles'' as ``bundles''. For any subsets $V$ of nodes, $N(V)$ is the set of all neighbors of the nodes in $V$. For a matching $M$, $V(M)$ is the set of vertices of $M$.
\section{Existence of Epistemic EFX Allocations}\label{sec:good-eefx}

In this section, we prove our main result that establishes existence of $\eefx$ allocations for any fair division instance with $n$ agents having monotone valuations. We start by proving an important structural property (in Lemma~\ref{lem:reduce}) that enables us to reduce an instance with lower number of agents. 

\begin{lemma}\label{lem:reduce}
    For any fair division instance, consider an agent $i \in \agents$ and $A \subseteq \items$ such that $A \notin \eefx^n_i(\items)$. Then for all bundles $B \in \eefx^{n-1}_i(\items \setminus A)$, we must have $B \in \eefx^n_i(\items)$.  
\end{lemma}
\begin{proof}
    %The proof is by contradiction. 
    For an agent $i \in \agents$, let us assume that bundle  $A \subseteq \items$ is such that $A \notin \eefx^n_i(\items)$. Now, consider any bundle $B \in \eefx^{n-1}_i(\items \setminus A)$, i.e., there exists an  $(n-1)$-certificate of $B$ for $i$ under $\items \setminus A$, we call it $C=\{C_1, \ldots, C_{n-2}\}$.
   %\paragraph{Case 1: \boldmath $i$ is a good agent.} 
    By definition, we have that 
    \begin{align} \label{eq:valid-red}
         v_i(B) \geq v_i(C_j \setminus g) \ \text{for all} \ j \in [n-2] \ \text{and} \ g \in C_j 
    \end{align}
   
     If $v_i(A) \geq v_i(B)$, then combining it with equation~(\ref{eq:valid-red}), we obtain $\{B, C_1, \ldots, C_{n-2}\}$ is an $n$-certificate of $A$ for $i$ under $\items$ and $A \in \eefx^n_i(\items)$, leading to a contradiction. Hence, we must have 
     \begin{align} \label{eq:valid-red2}
         v_i(B) \geq v_i(A)
     \end{align}
     Finally, combining equations~(\ref{eq:valid-red}) and (\ref{eq:valid-red2}), we obtain  $\{A, C_1, \ldots, C_{n-2}\}$ is an $n$-certificate of $B$ for $i$ under $\items$ and $B \in \eefx^n_i(\items)$. This completes our proof.
    %\paragraph{Case 2: \boldmath $i$ is a chore agent.} Let $c = \arg \max_{c' \in B} v_i(B \setminus c')$. For all $j \in [k-1]$, $B \setminus c \geq_i C_j$. If $A  >_i B \setminus c$, then $\{B, C_1, \ldots, C_{k-1}\}$ is a $k$-certificate of $A$ for $i$ under $S$ and $A \in \eefx^k_i(S)$ which is a contradiction. Thus, $B \setminus c \geq_i A$. Therefore, $\{A, C_1, \ldots, C_{k-1}\}$ is a $k$-certificate of $B$ for $i$ under $S$ and $B \in \eefx^k_i(S)$.
\end{proof}
Lemma \ref{lem:reduce} implies that if an agent $i$ finds a bundle $A$ to be $n$-epistemic-$\efx$ while no other agent finds $A$ to be $n$-epistemic-$\efx$, we can safely allocate $A$ to $i$, and remove $i$ and $A$ from the instance and find an $\eefx$ allocation of $\items \setminus A$ to the remaining $n-1$ agents. Note that we can repeat this process iteratively and remove $t \geq 1$ agents and $t$ bundles. The formal description is given in Corollary \ref{cor:reduce}.
\begin{corollary}[of Lemma \ref{lem:reduce}]\label{cor:reduce}
    For a fair division instance, consider a partial allocation $(X_{k+1}, X_{k+2}, \dots, X_n)$ to agents in the set $[n]\setminus[k]$. Let us assume that for all agents $i \in [n] \setminus [k]$ and all $j \in [k]$, we have $X_i \in \eefx^n_i$, $X_i \notin \eefx^n_j$. If $(X_1, \ldots, X_k)$ is an $\eefx$ allocation of $\items \setminus \bigcup_{\ell \in [k]}X_\ell$ for agents in $[k]$, then $(X_1, X_2, \ldots, X_n)$ is an $\eefx$ allocation for agents in $[n]$. 
\end{corollary}

\begin{algorithm}[t]
    \caption{$\alg=\eefx(\I)$}\label{alg:eefx}
    \textbf{Input:} A fair division instance $\I = (\agents, \items, \valus)$ where agent $i \in \agents = [n]$ has monotone valuation $v_i$ over the set of items $\items$\\
    \textbf{Output:} An allocation $X=(X_1, X_2, \ldots, X_n)$
    \begin{algorithmic}[1]
        \If{$\agents = \emptyset$}
            \State \Return $\emptyset$;
        \EndIf
        \State $n \leftarrow |\agents|$
        \State $(X_1, \ldots, X_n) \leftarrow$ an $\efx$ allocation of $\items$ among $n$ agents with valuation $v_n$;   
        \State $G \leftarrow$ $\eefx$-graph of $\{X_1, \ldots, X_n\}$;
        \State Let $M=\{(k+1,X_{k+1}), \ldots, (n,X_n)\}$ be a matching of size at least $1$ such that $N(\{X_{k+1}, \ldots, X_n\})=\{k+1, \ldots, n\}$;
        \State $\agents' \leftarrow [k]$;
        \State $\items' \leftarrow \items \setminus \bigcup_{\ell \in [n] \setminus [k]} X_\ell$;
        \State $\valus' \leftarrow (V_1, \ldots, V_k)$;
        \State $(X_1, \ldots, X_k) \leftarrow \eefx(\agents', \items', \valus')$;
        \State \Return $(X_1, X_2, \ldots, X_n)$;
    \end{algorithmic}
\end{algorithm}

 We will now give a high-level overview of our constructive proof for  establishing the existence of $\eefx$ allocation among arbitrary number of agents with monotone valuations using $\alg$ (see Algorithm \ref{alg:eefx}). For a fair division instance $\I = (\agents, \items, \valus)$, our algorithm $\alg$, starts by considering an $\efx$ allocation $(X_1, \dots, X_n)$ of $\items$ among $n$ agents with valuation $v_n$. We know such an allocation exists by the work of~\cite{plaut2020almost}. Next, we construct the $\eefx$-graph $G$ between the bundles $X_1, \dots, X_n$ and the agents. Lemma~\ref{lem:matching} proves that there will always exist a non-trivial matching\footnote{Without loss of generality, we can rename the bundles and agents in the matching $M$} $M=\{(k+1,X_{k+1}), \ldots, (n,X_n)\}$ such that $N(\{X_{k+1}, \ldots, X_n\})=\{k+1, \ldots, n\}$. That is, for every $j \in [n] \setminus [k]$, bundle $X_j \in \eefx^n_j(\items)$. 

Next, $\alg$ reduces the instance by removing the agents $\{k+1,k+2, \dots, n\}$ from $\agents$ with their bundles $X_{k+1}, X_{k+2}, \dots, X_n$ safely. Note that, no agent $i \in [k]$ has any edge in $G$ to any bundle $X_j$ for $j \in [n] \setminus [k]$. Finally, this also implies that finding an $\eefx$ allocation $(X_1, X_2, \dots,X_k)$ in the reduced instance and combining it with $(X_{k+1}, X_{k+2}, \dots, X_n)$ leads to an overall $\eefx$ allocation in the original instance. That is, our technique enables us to reduce our instance, find an $\eefx$ allocation in the reduced instance, and combine it in such a way that we produce an $\eefx$ allocation the the original instance.

\begin{figure}[t]
    \centering
    \includegraphics[width=5.5cm]{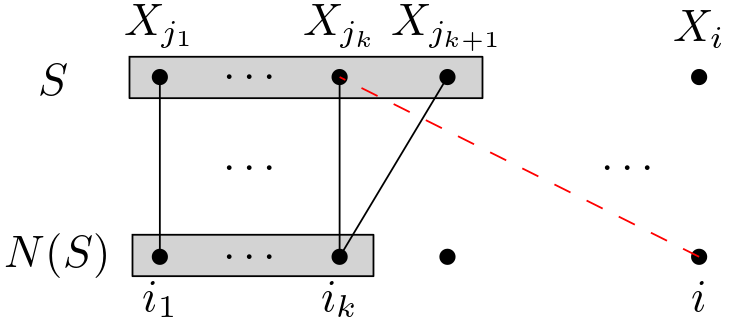}
    \caption{If $G(X)$ does not admit a perfect matching, then there exists a minimal subset $S=\{X_{j_1}, \ldots, X_{j_{k+1}}\}$ of bundles such that $|N(S)| < k+1$. Then, for all agent $i \in N(S)$ and all $\ell \in [k+1]$, no edge between $X_\ell$ and $i$ exists. In other words, no such red dashed edges can exist.}
    %Assume $N(S)=\{i_1, \ldots, i_k\}$. Then for all $i \notin \{i_1, \ldots, i_k\}$ and all $\ell \in [k+1]$, no edge between $X_\ell$ and $i$ exists. In other words, no such red dashed edges can exist.
    \label{fig:matching}
\end{figure}
We begin by proving Lemma~\ref{lem:matching}.
\begin{lemma}\label{lem:matching}
    For any fair division instance, consider an agent $i \in \agents$, let $(X_1, \ldots, X_n)$ be an $\efx$ allocation for an instance consisting of $n$ agents with identical valuations $v_i$. Let $G$ be the $\eefx$-graph with $n$ agents and $n$ bundles $X_1, \ldots, X_n$. Then there always exists a matching $M=\{(i_1,X_{j_1}), \ldots, (i_k,X_{j_k})\}$ of size at least $1$, such that $N(\{X_{j_1}, \ldots, X_{j_k}\})=\{i_1, \ldots, i_k\}$.
\end{lemma}
\begin{proof}
     To begin with, if $G$ has a perfect matching $M=\{(i_1,X_{j_1}), \ldots, (i_n,X_{j_n})\}$, then the lemma trivially holds true since $\agents \setminus \{i_1, \ldots, i_n\} = \emptyset$.
     
     Therefore, let us assume that no perfect matching exists in $G$. This implies that the Hall's condition is not satisfied, i.e., there exists a subset $S=\{X_{j_1}, \ldots, X_{j_{k+1}}\}$ of bundles such that $|N(\{X_{j_1}, \ldots, X_{j_{k+1}}\})| < k+1$. See Figure \ref{fig:matching} for a better intuition. We assume that the subset $S=\{X_{j_1}, \ldots, X_{j_{k+1}}\}$ is minimal. That is, for all $S' \subsetneq S$, we have $N(S') \geq |S'|$. Now consider $T=\{X_{j_1}, \ldots, X_{j_k}\} \subsetneq S$. By minimality of $S$, we know that Hall's condition holds for $T$, i.e., there exists a perfect matching, say $M=\{(i_1,X_{j_1}), \ldots, (i_k,X_{j_k})\}$ between the nodes in $T$ and $N(T)$. Since $|N(S)|<k+1$ and $\{i_1, \ldots, i_k\} \subseteq N(T) \subseteq N(S)$, it follows that $N(S) = N(T) = \{i_1, \ldots, i_k\}$. 
     
     Note that since $(X_1, \dots, X_n)$ is an $\efx$ allocation for an instance with identical valuations $v_i$, we know that $i \in N(S)$, thus $k \geq 1$. Hence, $M=\{(i_1,X_{j_1}), \ldots, (i_k,X_{j_k})\}$ is a matching of size $k \geq 1$, such that $N(\{X_{j_1}, \ldots, X_{j_k}\})=\{i_1, \ldots, i_k\}$.  The stated claim stands proven.
\end{proof}
\begin{theorem}[\cite{plaut2020almost}]\label{thm:PR}
    When agents have identical monotone valuations, there always exists an $\efx$ allocation.
\end{theorem}

We are now ready to discuss our main result that constructively establishes the existence of $\eefx$ allocation among arbitrary number of agents with monotone valuations using $\alg$.

\begin{theorem}\label{thm:eefx}
    $\eefx$ allocations exist for any fair division instance with monotone valuations.  In particular, $\alg$ returns an $\eefx$ allocation.
\end{theorem}

\begin{proof}    
    We begin by proving that Algorithm \ref{alg:eefx} terminates. By Lemma \ref{lem:matching}, a matching $M=\{(i_1,X_{i_1}), \ldots, (i_t,X_{i_t})\}$ of size at least $1$ exists such that $N(\{X_{i_1}, \ldots, X_{i_t}\})=\{i_1, \ldots, i_t\}$. Note that we can rename the bundles and the agents and without loss of generality assume that the considered matching is $M=\{(k+1,X_{k+1}), \ldots, (n,X_n)\}$. Therefore, after removing $\{k+1, \ldots, n\}$ from $\agents$, the size of $\agents$ decreases. Hence, the depth of the recursion is bounded by $n$ (the initial number of agents). 
    
    We prove the correctness of $\alg$ by using induction on the number of the agents. If $\agents=\emptyset$, then $\emptyset$ is an $\eefx$ allocation. We assume that $\alg$ returns an $\eefx$ allocation for any fair division instance with $n' < n$ agents with monotone valuations. Consider the matching $M$ described in $\alg$. We will show the output allocation of $\alg$ for $n$ agents is $\eefx$ as well. For any $i \in [n] \setminus [k]$ and any $j \in [k]$, the matching $M$ ensures that we have $X_i \in \eefx^n_i$, and $X_i \notin \eefx^n_j$ (see Figure~\ref{fig:matching}). By induction hypothesis $(X_1, \ldots, X_k)$ is an $\eefx$ allocation of $\items \setminus \bigcup_{\ell \in [k]}X_\ell$ for agents in $[k]$. Thus, by Corollary \ref{cor:reduce}, $(X_1, X_2, \ldots, X_n)$ is an $\eefx$ allocation for agents in $[n]$.
\end{proof}

\begin{remark}
    All proofs of this section that we have for the setting when items are goods easily extend to the setting when these items are `\emph{chores}'. Formally, when agent valuations are monotonically decreasing, then EEFX allocations are guaranteed to exist for an arbitrary number of agents.
    %In the spirit of not repeating the proofs, we omit them here.
\end{remark}
\section{Hardness Results} \label{sec:hardness}

In this section, we complement our existential result of $\eefx$ allocations for monotone valuations by proving computational and information-theoretic lower bounds for finding an  $\eefx$ allocation.
When agents have submodular valuation functions, the way to compute the value $v(S)$ for a subset $S$ of the items is through making value queries.~\cite{plaut2020almost} proved that exponentially many value queries are required to compute an $\efx$ allocation \emph{even} for two agents with identical submodular valuations. Formally, they proved the following information-theoretic lower bounds.
\begin{theorem}[\cite{plaut2020almost}]\label{thm:PR-exp-lower}
    The query complexity of finding an $\efx$ allocation with $|\items|=2k+1$ many items is $\Omega(\frac{1}{k} {2k+1 \choose k})$, even for two agents with identical submodular valuations.
\end{theorem}

Moreover, \cite{Goldberg2023} proved the following computational hardness for $\efx$ allocations.
\begin{theorem}[\cite{Goldberg2023}]\label{thm:cite-pls-hard}
    The problem of computing an $\efx$ allocation for two agents with identical submodular valuations is $\pls$-complete.
\end{theorem}
We ask our readers to refer to Appendix \ref{app:pls} for further discussion on the complexity class $\pls$. Let us now define the computational problems corresponding to finding $\efx$ and $\eefx$ allocations. 
\begin{definition}($\idefx$)
Given a fair division instance $\I=([2], \M, (v,v))$ with two agents having identical submodular valuations $v$, find an $\efx$ allocation.
\end{definition}

\begin{definition}($\ideefx$)
Given a fair division instance $\I=([n], \M, (v, \dots, v))$ with $n$ agents having identical submodular valuations $v$, find an $\eefx$ allocation.
\end{definition}

We reduce the problem of finding an $\efx$ allocation for two agents with identical submodular valuations ($\idefx$) to finding an  $\eefx$ allocation for an arbitrary number of agents with identical submodular valuations ($\ideefx$), thereby establishing similar hardness results for the latter. 
%It is also relevant to note that our reduction works even for three agents, and hence Theorems~\ref{thm:exp} and \ref{thm:pls} hold true even for the problem of computing $\eefx$ allocations even for three agents with identical submodular valuations. 

\textbf{Our Reduction:}
Consider an arbitrary instance $\I=([2], \M, (v,v))$ of $\idefx$ with two agents having identical submodular valuations $v$. Let $\I' = ([n], \M', (v', \ldots, v')\})$ be an instance of $\ideefx$ with $n$ agents having identical valuations $v'$ over the set of items $\M' =  \M \cup \{h_1, \ldots, h_{n-2}\}$.
We define the valuation $v$ as follows.
\begin{itemize}
    \item For all $S \subseteq \M$, $v'(S)=v(S)$. 
    \item For all $j \in [n-2]$, $v'(h_j) = 2v(\items)+1$.
    \item For all $j \in [n-2]$ and $S \subseteq \M' \setminus \{h_j\}$, $v'(S \cup \{h_j\}) = v(S)+v(h_j)$.
\end{itemize}
We call items $h_1, \ldots, h_{n-2}$ heavy items. Note that we can compute $\I'$ from $\I$ in polynomial time.

\begin{restatable}{lemma}{subLem}\label{lem:submodul}
    If $v$ is a submodular function, then $v'$ is a submodular function as well.
\end{restatable}
\begin{proof}
    We need to prove that for all $S, T \subseteq \items$, $v'(S)+v'(T) \geq v'(S \cup T) + v'(S \cap T)$.
    Let $H_S$ and $H_T$ be the set of all heavy items in $S$ and $T$ respectively. We have 
    \begin{align*}
        v'(S)+v'(T) &= v'(S \setminus H_S) + v'(H_S) + v'(T \setminus H_T) + v'(H_T) \\
        &= \left( v(S \setminus H_S) + v(T \setminus H_T) \right) +  v'(H_S) + v'(H_T) ) \\
        &\geq v((S \setminus H_S) \cup (T \setminus H_T)) + v((S \setminus H_S) \cap (T \setminus H_T)) + v'(H_S) + v'(H_T) \tag{submodolarity of $v$} \\
        &= v'((S \cup T) \setminus (H_S \cup H_T)) + v'((S \cap T) \setminus (H_S \cap H_T)) + v'(H_S) + v'(H_T) \\
        &= v'((S \cup T) \setminus (H_S \cup H_T)) + v'((S \cap T) \setminus (H_S \cap H_T)) \\
        &\textcolor{white}{=} + v'(H_S \cup H_T) + v'(H_S \cap H_T) \tag{additivity of $v'$ on heavy items}\\
        &= v'(S \cup T) + v'(S \cap T).
    \end{align*}
\end{proof}

\begin{lemma}\label{lem:poly-reduction}
    Given any $\eefx$ allocation $A$ in $\I'$, we can create an $\efx$ allocation in $\I$ in polynomial time, where $\I$ and $\I'$ are as defined above.     
\end{lemma}
\begin{proof}
    Let us assume that $A = (A_1, \dots, A_n)$ is an $\eefx$ allocation in instance $\I'$.
    To begin with, note that there are $n-2$ heavy items in $\I'$, and hence, by pigeonhole principle, there exists at least two agents, say $i,j \in \agents'$ such that they receive no heavy item under $A$. Without loss of generality, let us assume that $i=1$ and $j=2$, and hence we have $A_1, A_2 \subseteq \items$. This implies that we have
    \begin{align} \label{eq:heavy}
       v'(A_1)=v(A_1), v'(A_2)=v(A_2), \text{and} \ v(A_1), v(A_2) < 2v(\items)+1
    \end{align}
    Without loss of generality, let us assume $v(A_2) \geq v(A_1)$.
    
    We will prove $( A_1, \M \setminus A_1 )$ forms an $\efx$ allocation in $\I$. Note that valuations $v$ and $v'$ coincide for the bundles $A_1$ and $\M \setminus A_1$. Since $A$ is $\eefx$ in $\I'$, let us denote the $n$-certificate for agent $1$ with respect to $A_1$ by $C=(C_2, C_3, \ldots, C_n)$. First, we prove that no bundle $C_k$ with a heavy item can have any other item as well. Assume otherwise. Let $\{g,h_j\} \subseteq C_k$ for some $k \in \{2,\ldots,n\}$ and some $j \in [n-2]$ and $g \neq h_j$. Then, we have 
    \begin{align*}
        v'(C_k \setminus \{g\}) \geq v'(h_j) = 2v(M)+1 > v'(A_1)
    \end{align*}
     where, the last inequality uses equation~(\ref{eq:heavy}). This implies that agent $1$ strongly envies bundle $C_k$ which is a contradiction to our assumption that $C$ forms an $n$-certificate for bundle $A_1$ in instance $I'$. Therefore, the $n-1$ bundles in the $n$-certificate must look like $\{C_2, \ldots, C_n\} = \{\{h_1\}, \ldots, \{h_{n-2}\}, \M \setminus A_1\}$. First, note that, agent $1$ with bundle $A_1$ must not strongly envy bundle $\M \setminus A_1$ since $C$ is an $n$-certificate. And since, $\M \setminus A_1= A_2$, we already have $v(A_2)\geq v(A_1)$.   Therefore, the allocation $(A_1, \M \setminus A_1)$ forms an $\efx$ allocation in $\I$.
    % Note the since $A_2 \subseteq \M \setminus A_1$, $v(\M \setminus A_1) \geq v(A_2) \geq v(A_1)$.
\end{proof}

\begin{theorem} \label{thm:exp}
    The query complexity of the $\eefx$ allocation problem with $|\items|=2k+n-1$ many items is $\Omega(\frac{1}{k} {2k+1 \choose k})$, for arbitrary number of agents $n$ with identical submodular valuations.
\end{theorem}
\begin{proof}
    Consider any arbitrary instance $\I=([2], \M, (v,v))$ with two agents having identical submodular valuations $v$ and $|\M| = 2k+1$ items. Create the instance $\I'$ as described above. Using Lemma~\ref{lem:submodul}, $\I'$ consists of $n$ agents with identical submodular valuations. By Lemma \ref{lem:poly-reduction}, given any $\eefx$ allocation $A$, we can obtain an $\efx$ allocation for $\I$ in polynomial time. Finally, using Theorem \ref{thm:PR-exp-lower}, we know that the query complexity of finding an $\efx$ allocation in $\I$ is  $\Omega(\frac{1}{k} {2k+1 \choose k})$. Hence, the query complexity of $\eefx$ for $n$ agents with identical submodular valuations admits the same lower bound. This establishes the stated claim.
\end{proof}

Finally, our next result follows using Lemma \ref{lem:poly-reduction} and Theorem \ref{thm:cite-pls-hard}.
\begin{theorem} \label{thm:pls}
    The problem of computing an $\eefx$ allocation for arbitrary number of agents with identical submodular valuations is $\pls$-hard.
\end{theorem}

Since our reduction work even for three agents, Theorems~\ref{thm:exp} and \ref{thm:pls} hold true for the problem of computing $\eefx$ allocations even for three agents with identical submodular valuations. Note that the set of $\efx$ and $\eefx$ allocations coincide for the case of two agents and hence it inherits the same computational hardness guarantees as that of $\efx$ here.

\section{Conclusion and Open Problems} \label{sec:conc}

In this work, we establish the existence of $\eefx$ allocations for an arbitrary number of agents with general monotone valuations. To the best of our knowledge, this is the strongest existential guarantee known for any relaxation of $\efx$ in discrete fair division. Furthermore, we also prove that the problem of computing an $\eefx$ allocation for instances with an arbitrary number of agents with submodular valuations is $\pls$-hard and requires an exponential number of valuations queries as well. Our existential result of $\eefx$ allocations for monotone valuations has opened a variety of major problems in discrete fair division. We list three of them here, that we believe should be explored first.

The first interesting question is, for submodular or monotone valuations, explore the possibility of a PTAS for computing an $\eefx$ allocation, or otherwise prove its $\apx$-hardness. An equally exciting problem would be to explore the compatibility of $\eefx$ and $\ef$1 allocations. Even for instances with additive valuations, does there always exist an allocation that is simultaneously both $\eefx$ and $\ef$1? If yes, can we compute it? What about similar compatibility question of $\eefx$ with \emph{Nash social welfare}? We know that a \emph{maximum Nash welfare} (MNW) allocation is both $\ef$1 and \emph{Pareto-optimal} \cite{caragiannis2016unreasonable}. What kind of a relation\footnote{We know that MNW allocations may not be $\eefx$. Consider the following example with two agents and three items $\{a,b,c\}$. Agent $1$ values item $a$ at $10$, item $b$ at $1$, and item $c$ at $\varepsilon>0$. Agent $2$ values item $a$ at $10$, item $b$ at $\varepsilon$, and item $c$ at $1$. Here, the unique MNW allocation (that assigns items $a$ and $b$ to agent $1$, and item $c$ to agent $2$) is not $\eefx$.} exist between $\eefx$ and MNW allocations?

\bibliographystyle{plainnat}
\bibliography{references}
\newpage
\appendix
\iffalse
\section{Missing Proofs}\label{appendix}
\subLem*
\begin{proof}
    We need to prove that for all $S, T \subseteq \items$, $v'(S)+v'(T) \geq v'(S \cup T) + v'(S \cap T)$.
    Let $H_S$ and $H_T$ be the set of all heavy items in $S$ and $T$ respectively. We have 
    \begin{align*}
        v'(S)+v'(T) &= v'(S \setminus H_S) + v'(H_S) + v'(T \setminus H_T) + v'(H_T) \\
        &= \left( v(S \setminus H_S) + v(T \setminus H_T) \right) +  v'(H_S) + v'(H_T) ) \\
        &\geq v((S \setminus H_S) \cup (T \setminus H_T)) + v((S \setminus H_S) \cap (T \setminus H_T)) + v'(H_S) + v'(H_T) \tag{submodolarity of $v$} \\
        &= v'((S \cup T) \setminus (H_S \cup H_T)) + v'((S \cap T) \setminus (H_S \cap H_T)) + v'(H_S) + v'(H_T) \\
        &= v'((S \cup T) \setminus (H_S \cup H_T)) + v'((S \cap T) \setminus (H_S \cap H_T)) \\
        &\textcolor{white}{=} + v'(H_S \cup H_T) + v'(H_S \cap H_T) \tag{additivity of $v'$ on heavy items}\\
        &= v'(S \cup T) + v'(S \cap T).
    \end{align*}
\end{proof}
\fi

\section{Polynomial Local Search ($\pls$)}\label{app:pls}

The following description of the complexity class $\pls$ is taken from Section 7.2 in \cite{ABR22}.

The class $\pls$ (Polynomial Local Search) was defined by~\citet{johnson1988easy} to capture the complexity of finding local optima of optimization problems. Here, a generic instance $\I$ of an optimization problem has a corresponding finite set of solutions $S(\I)$ and a potential $c(s)$ associated with each solution $s \in S(\I)$. The objective is to find a solution that maximizes (or minimizes) this potential. In the local search version of the problem, each solution $s \in S(\I)$ additionally has a well-defined neighborhood $N(s) \in 2^{S(\I)}$ and the objective is to find a local optimum, i.e., a solution $s\in S(\I)$ such that no solution in its neighborhood $N(s)$ has a higher potential.

\begin{definition}[$\pls$]
	\label{definition:pls}
	Consider an optimization problem $\mathcal{X}$, and for all input instances $\I$ of $\mathcal{X}$ let $S(\mathcal{I})$ denote the finite set of feasible solutions for this instance, $N(s)$ be the neighborhood of a solution $s \in S(\I)$,  and $c(s)$ be the potential of solution $s$. The desired output is a local optimum with respect to the potential function. %$N: S \mapsto 2^S$ is a neighbourhood function, $c: S \mapsto \mathbb{R}$ is a potential function and the desired output is a local optimum with respect to the potential (either maximum or minimum).  
	
	Specifically, $\mathcal{X}$ is a polynomial local search problem (i.e., $\mathcal{X} \in \pls$) if all solutions are bounded in the size of the input $\mathcal{I}$ and there exists polynomial-time algorithms $\mathcal{A}_1$, $\mathcal{A}_2$, and $\mathcal{A}_3$ such that:
	\begin{enumerate}
		\item $\mathcal{A}_1$ tests whether the input $\mathcal{I}$ is a legitimate instance of $\mathcal{X}$ and if yes, outputs a solution $s_{\text{initial}} \in S(\mathcal{I})$.
		\item $\mathcal{A}_2$ takes as input instance $\mathcal{I}$ and candidate solution $s$, tests if $s \in S(\mathcal{I})$ and if yes, computes $c(s)$.
		\item $\mathcal{A}_3$ takes as input instance $\mathcal{I}$ and candidate solution $s$, tests if $s$ is a local optimum and if not, outputs $s' \in N(s)$ such that $c(s') > c(s)$ (the inequality is reversed for the minimization version).
	\end{enumerate}
\end{definition}

Each $\pls$ problem comes with an associated local search algorithm that is implicitly described by the three algorithms mentioned above. The first algorithm is used to find an initial solution to the problem and the third algorithm is iteratively used to find a potential-improving neighbor until a local optimum is reached.

\end{document}